\documentclass[12pt, a4paper]{article}

\usepackage[utf8]{inputenc}
\usepackage[T1]{fontenc}
\usepackage{amsmath, amssymb, amsthm}
\usepackage{mathtools}
\usepackage{graphicx}
\usepackage{booktabs}
\usepackage{hyperref}
\usepackage[margin=2.5cm]{geometry}
\usepackage{natbib}
\usepackage{algorithm}
\usepackage{algpseudocode}
\usepackage{float}
\usepackage{caption}
\usepackage{subcaption}

\newtheorem{theorem}{Theorem}
\newtheorem{proposition}[theorem]{Proposition}
\newtheorem{lemma}[theorem]{Lemma}
\newtheorem{corollary}[theorem]{Corollary}

\newtheorem{remark}{Remark}

\newcommand{\E}{\mathbb{E}}
\newcommand{\Var}{\mathrm{Var}}
\newcommand{\Cov}{\mathrm{Cov}}
\newcommand{\R}{\mathbb{R}}
\newcommand{\dd}{\mathrm{d}}
\DeclareMathOperator*{\argmax}{arg\,max}

\title{Optimal Hedge Ratio for Delta-Neutral Liquidity Provision \\
under Liquidation Constraints}

\author{Atsushi Hane\thanks{Independent researcher. Email: \texttt{atsushihane@gmail.com}}}

\date{March 2026}

\begin{document}
\maketitle

\begin{abstract}
We study the problem of optimally hedging the price exposure of
liquidity positions in constant-product automated market makers (AMMs) when the hedge is
funded by collateralized borrowing.  A liquidity provider (LP) who
borrows tokens to construct a delta-neutral position faces a
trade-off: higher hedge ratios reduce price exposure but increase
liquidation risk through tighter collateral utilization.  We model
token prices as correlated geometric Brownian motions and derive the
hedge ratio $h$ that maximizes risk-adjusted return subject to a
liquidation-probability constraint expressed via a first-passage-time
bound.  The unconstrained optimum $h^{*}$ admits a closed-form
expression, but at $h^{*}$ the liquidation probability is
prohibitively high.  The practical optimum
$h^{**} = \min(h^{*},\,\bar{h}(\alpha))$ is determined by the
binding liquidation constraint $\bar{h}(\alpha)$, which we evaluate
analytically via the first-passage-time formula and confirm with
Monte Carlo simulation.  Simulations calibrated to on-chain data validate
the analytical results, demonstrate robustness across realistic
parameter ranges, and show that the optimal hedge ratio
lies between 50\% and 70\% for typical DeFi lending conditions.
Practical guidelines for rebalancing frequency and position sizing are
also provided.


\end{abstract}

\section{Introduction}\label{sec:intro}


Automated market makers (AMMs) such as Uniswap~v2 \citep{uniswap2020}
have become a cornerstone of decentralized finance
(DeFi; see \citealt{xu2023} for a comprehensive survey), enabling
permissionless token exchange via constant-product liquidity pools.
Liquidity providers earn trading fees but are exposed to
\emph{impermanent loss} (IL)---the shortfall relative to a simple
buy-and-hold portfolio when token prices diverge.%
\footnote{\citet{milionis2022} argue that the relevant LP cost
metric is \emph{fees\,--\,LVR} rather than IL\@.  Our framework
hedges pure price exposure, not non-linear costs such as IL or
LVR; see Remark~3.}

A natural mitigation is to \emph{borrow} the constituent tokens from a
lending protocol and sell them, creating a short position that offsets
the LP's long exposure.  This ``delta-neutral'' approach is widely
discussed in practitioner forums, yet a rigorous treatment of the
resulting \emph{liquidation risk} is absent from the literature.

The contribution of this paper is threefold:
\begin{enumerate}
    \item We formalize the LP hedging problem as a stochastic
          optimization over the hedge ratio~$h$, incorporating
          borrowing costs, reward income, and a hard
          loan-to-value (LTV) constraint.
    \item We derive the unconstrained optimal hedge ratio
          $h^{*}$ in closed form and show that the practical
          optimum $h^{**}$ is determined by the binding
          liquidation constraint, whose threshold $\bar{h}(\alpha)$
          is characterized via a first-passage-time approximation.
    \item We validate the model with Monte Carlo simulation
          calibrated to on-chain pool and lending data, and provide
          actionable rebalancing rules.
\end{enumerate}

The remainder of the paper is organized as follows.
Section~\ref{sec:related} reviews related work.
Section~\ref{sec:model} presents the model.
Section~\ref{sec:optimal} derives the optimal hedge ratio.
Section~\ref{sec:liquidation} analyzes liquidation probability using
first-passage-time theory.
Section~\ref{sec:simulation} describes the Monte Carlo study.
Section~\ref{sec:results} reports the results.
Section~\ref{sec:discussion} discusses practical implications.
Section~\ref{sec:conclusion} concludes.

\section{Related Work}\label{sec:related}

\paragraph{Impermanent loss.}
The concept of IL in constant-product AMMs was first described
informally by \citet{pintail2019} and formalized in the
Uniswap~v2 whitepaper \citep{uniswap2020}.  \citet{aigner2021}
quantified IL under geometric Brownian motion and derived
closed-form expressions for expected IL as a function of
volatility.  \citet{clark2020} provided a similar analysis with
emphasis on replication strategies.  \citet{loesch2021} extended
the analysis to concentrated-liquidity AMMs (Uniswap~v3).
\citet{angeris2020} established rigorous theoretical foundations for
constant-product markets, proving that AMM prices must closely track
external market prices under no-arbitrage conditions.
\citet{milionis2022} introduce \emph{loss-versus-rebalancing}
(LVR), a distinct, path-dependent measure of the adverse-selection
cost that LPs bear by offering a continuous trading option to
arbitrageurs; while related to IL, LVR is monotonically increasing
and captures costs that IL---being path-independent---does not.  \citet{evans2022} derive optimal
fee structures to compensate LPs for adverse selection, while
\citet{heimbach2022} provide empirical evidence that
maximal-extractable-value (MEV) activity amplifies LP losses
beyond what IL alone predicts.

\paragraph{Delta-neutral strategies.}
\citet{khakhar2022} propose a delta-hedging algorithm for LP
positions using options on the underlying tokens, but do not derive
an optimal hedge ratio or model liquidation risk from the hedging
instrument itself.  \citet{lipton2024} develop a unified framework
for hedging LP price exposure via static option replication and dynamic
delta-hedging, covering both Uniswap~v2 and v3; their approach
assumes access to liquid options markets, which remain largely
unavailable in DeFi.
\citet{capponi2021} study optimal portfolio allocation in
decentralized exchanges, and \citet{fan2022} analyze differential
liquidity provision in concentrated-liquidity AMMs.  To the best of
our knowledge, no prior work explicitly models the hedge ratio as a
continuous decision variable subject to a liquidation constraint
arising from collateralized borrowing.

\paragraph{Liquidation risk in DeFi lending.}
\citet{qin2021} and \citet{perez2021} study liquidation cascades
in lending protocols from a systemic-risk perspective.
\citet{gudgeon2020} analyze the design of DeFi lending protocols
and the role of liquidation mechanisms.
\citet{bartoletti2021} provide a systematic classification of DeFi
protocols and their composability risks, including cascading
liquidations across lending and AMM layers.  Our work differs by
focusing on the \emph{individual} LP's problem of choosing a hedge
ratio that balances variance reduction against liquidation probability.

\paragraph{First passage time.}
The first-passage-time theory for Brownian motion is classical
\citep{karatzas1991}.  We apply it to the LTV process via a
moment-matching approximation for the sum of correlated geometric
Brownian motions, following the approach of \citet{milevsky1998}.

\section{Model}\label{sec:model}

\subsection{Price dynamics}

Let $S_t^A$ and $S_t^B$ denote the prices of tokens $A$ and $B$ at
time~$t$, measured in a num\'eraire (e.g., USDC).  We model prices as
correlated geometric Brownian motions:
\begin{align}
    \frac{\dd S_t^A}{S_t^A} &= \mu_A \, \dd t
        + \sigma_A \, \dd W_t^A, \label{eq:sde_a} \\
    \frac{\dd S_t^B}{S_t^B} &= \mu_B \, \dd t
        + \sigma_B \, \dd W_t^B, \label{eq:sde_b}
\end{align}
where $\dd W_t^A \, \dd W_t^B = \rho \, \dd t$ and
$\rho \in (-1, 1)$ is the instantaneous correlation.

\subsection{Constant-product AMM}

A constant-product pool maintains the invariant
$x_t \cdot y_t = L^2$, where $x_t$ and $y_t$ are the reserves of
tokens $A$ and $B$, and $L$ is the liquidity parameter.  An LP who
deposits value $V_0$ at time~0 holds a claim whose \emph{mark-to-market
value} (excluding accumulated trading fees and farming rewards) at
time~$t$ is
\begin{equation}\label{eq:lp_value}
    V_t^{\text{LP}} = V_0 \sqrt{\frac{S_t^A}{S_0^A}
        \cdot \frac{S_t^B}{S_0^B}}.
\end{equation}
Fee and reward income is accounted for separately via the cumulative
reward term $R_t$ in \eqref{eq:portfolio}.

\subsection{Hedging via borrowing}

The LP deposits collateral $C$ (in stablecoin) into a lending protocol
and borrows fractions of tokens $A$ and $B$.  Define the
\emph{hedge ratio} $h \in [0, 1]$ as the fraction of LP token
exposure that is offset by the borrowed position.  At inception the LP
borrows
\begin{equation}\label{eq:borrow}
    D_0^A = h \cdot \frac{V_0}{2 S_0^A}, \qquad
    D_0^B = h \cdot \frac{V_0}{2 S_0^B}
\end{equation}
tokens $A$ and $B$, and uses them, together with spot-purchased tokens,
to fund the LP position.  The borrow rates are $r_A$ and $r_B$ per annum.

\subsection{Portfolio value}

The LP funds the position with a mix of equity and borrowed tokens.
Equity capital deployed is $C + (1-h)V_0$: the collateral~$C$ (in
stablecoin, deposited in the lending protocol) plus $(1-h)V_0$ used
to purchase the unhedged fraction of LP tokens.  The borrowed
tokens go directly into the LP alongside the spot-purchased ones.

The net portfolio value at time~$t$ (before rebalancing) is
\begin{equation}\label{eq:portfolio}
    \Pi_t = \underbrace{V_t^{\text{LP}} + R_t}_{\text{LP + rewards}}
        + \underbrace{C(1 + r_f t)}_{\text{collateral + yield}}
        - \underbrace{\frac{h V_0}{2}
            \!\left(\frac{S_t^A}{S_0^A}
            + \frac{S_t^B}{S_0^B}\right)}_{\text{debt value}}
        - \underbrace{\frac{h V_0}{2}(r_A + r_B)\,t}_{\text{borrow cost}},
\end{equation}
where $R_t$ denotes cumulative LP rewards (fees + incentives),
growing linearly as $R_t = (R/V_0)\,V_0\,t$ with
$R/V_0$ the annualized reward-to-value ratio (Table~\ref{tab:calibration}),
and $r_f$ is the stablecoin supply rate earned on collateral.

\subsection{Loan-to-value constraint}

The lending protocol liquidates the position if the loan-to-value
ratio exceeds a threshold $\ell_{\max}$:
\begin{equation}\label{eq:ltv}
    \text{LTV}_t = \frac{\frac{h V_0}{2}
        \bigl(\frac{S_t^A}{S_0^A} + \frac{S_t^B}{S_0^B}\bigr)}
        {C} \leq \ell_{\max}.
\end{equation}
The numerator (debt in USD) grows when token prices rise, while the
denominator (stablecoin collateral) is fixed.  Note that
$\text{LTV}_0 = hV_0/C$, so higher hedge ratios start with higher
leverage.

\begin{remark}[Accrued interest]\label{rem:accrued}
Equation~\eqref{eq:ltv} omits accrued borrow interest from the
numerator.  Including it would add a term
$\frac{h V_0}{2}(r_A + r_B)\,t$ to the debt, raising LTV by
$\frac{h V_0 (r_A + r_B)\,t}{2C}$.  At the calibrated values
($h = 0.60$, $r_A + r_B = 0.18$, $T = 0.25$, $C/V_0 = 2.0$),
this amounts to ${\approx}\,0.7$\,pp of LTV---negligible relative
to the 50\,pp buffer between initial LTV (30\%) and
$\ell_{\max}$ (80\%).  The Monte Carlo simulation
(Section~\ref{sec:simulation}) does incorporate accrued interest
in the daily LTV check, so all numerical results account for
this effect; only the analytical expressions in
Sections~\ref{sec:optimal} and~\ref{sec:liquidation} use the
simplified form~\eqref{eq:ltv}.
\end{remark}

\section{Optimal Hedge Ratio}\label{sec:optimal}

Throughout this section we set $\mu_A = \mu_B = 0$ in
\eqref{eq:sde_a}--\eqref{eq:sde_b}.  This is \emph{not} an
appeal to risk-neutral pricing (no derivative is being priced);
rather, it reflects the assumption that the LP holds no
directional view on token prices and regards reward income as the
sole source of expected profit.  The zero-drift assumption is analytically convenient and
conservative in bull markets: any positive drift in token prices
would increase debt value and thus worsen liquidation risk, making
our estimates an upper bound on feasible hedge ratios.  Conversely,
in bear markets ($\mu < 0$), falling prices reduce debt value and
improve LTV, so the short leg generates profit and the liquidation
constraint relaxes.  The optimal $h^{**}$ would then increase,
meaning the zero-drift estimates are conservative in both
directions: they understate the feasible hedge ratio in bear
markets and overstate it in bull markets.

\subsection{Objective function}

We adopt the Sharpe ratio as the objective for analytical
tractability and comparability with the traditional hedging
literature.  While the return distribution exhibits a mass
point at the liquidation loss (visible in
Table~\ref{tab:results_hedge} as the discontinuous VaR jump
between $h = 0.70$ and $h = 0.80$), the 5\% VaR reported
alongside confirms that the Sharpe-optimal $h^{**}$ also
performs well under tail-risk criteria.  A formal CVaR
optimization is left for future work.

We maximize the Sharpe ratio of the dollar profit and loss
$\Pi_T - \Pi_0$ over the hedge ratio~$h$:
\begin{equation}\label{eq:objective}
    h^{*} = \argmax_{h \in [0,1]} \;
        \frac{\E[\Pi_T - \Pi_0]}
             {\sqrt{\Var(\Pi_T)}}
    \quad \text{s.t.} \quad
    \Pr\!\bigl(\exists\, t \in [0,T]:
        \text{LTV}_t > \ell_{\max}\bigr) \leq \alpha.
\end{equation}
We first solve the unconstrained problem (ignoring the liquidation
constraint) and then incorporate the constraint in
Section~\ref{sec:liquidation}.

\begin{remark}[Dollar vs.\ ROE-based Sharpe ratio]\label{rem:dollar_roe}
The analytical derivation maximizes the Sharpe ratio of dollar
P\&L $\Pi_T - \Pi_0$.  The Monte Carlo analysis
(Section~\ref{sec:simulation}) instead reports the Sharpe ratio
of ROE $:= (\Pi_T - \Pi_0)/\Pi_0$, where
$\Pi_0 = C + (1-h)V_0$ (Eq.~\ref{eq:equity}).  Since $\Pi_0$
depends on~$h$, the two objectives are not identical.  However,
for $C/V_0 = 2.0$ the equity base varies by at most a factor
of $3.0/2.0 = 1.5$ across $h \in [0,1]$, and both criteria
agree on the location of $h^{**}$ because the practical optimum
is determined by the binding liquidation constraint
(Remark~\ref{rem:interior_from_constraint}), not by the
unconstrained FOC.
\end{remark}

\subsection{Lognormal moment lemma}

The following identity is used repeatedly.

\begin{lemma}\label{lem:mgf}
Let $X_T = \ln(S_T^A/S_0^A)$ and $Y_T = \ln(S_T^B/S_0^B)$.
Under zero-drift GBM ($\mu_A = \mu_B = 0$), the joint
moment generating function is
\begin{equation}\label{eq:mgf}
    \E\!\left[e^{a X_T + b Y_T}\right]
    = \exp\!\left(
        \frac{a(a-1)\sigma_A^2 + b(b-1)\sigma_B^2
              + 2ab\,\rho\,\sigma_A\sigma_B}{2}\,T
    \right)
\end{equation}
for all $a, b \in \R$.
\end{lemma}
\begin{proof}
Under zero-drift GBM, $X_T \sim \mathcal{N}(-\tfrac{1}{2}\sigma_A^2 T,\;
\sigma_A^2 T)$ and $Y_T \sim \mathcal{N}(-\tfrac{1}{2}\sigma_B^2 T,\;
\sigma_B^2 T)$ with $\Cov(X_T, Y_T) = \rho\sigma_A\sigma_B T$.
The linear combination $aX_T + bY_T$ is normal with mean
$-(a\sigma_A^2 + b\sigma_B^2)T/2$ and variance
$(a^2\sigma_A^2 + b^2\sigma_B^2 + 2ab\rho\sigma_A\sigma_B)T$.
Applying the lognormal MGF $\E[e^Z] = e^{\mu_Z + \sigma_Z^2/2}$
and simplifying yields \eqref{eq:mgf}.
\end{proof}

\noindent
Table~\ref{tab:moments} collects the moments used below, each
obtained by substituting the appropriate $(a,b)$ into
Lemma~\ref{lem:mgf}.

\begin{table}[H]
\centering
\caption{Key lognormal moments under zero-drift GBM
($\mu_A = \mu_B = 0$, $p_i := S_T^i/S_0^i$).}
\label{tab:moments}
\begin{tabular}{@{}llll@{}}
\toprule
Quantity & $(a,b)$ & Expression \\
\midrule
$\E[p_A]$ & $(1,0)$ & $1$ \\
$\E[p_B]$ & $(0,1)$ & $1$ \\
$\E[\sqrt{p_A p_B}]$ & $(\tfrac{1}{2},\tfrac{1}{2})$
    & $e^{-\phi T}$, \;
      $\phi := \dfrac{\sigma_A^2 + \sigma_B^2
                       - 2\rho\sigma_A\sigma_B}{8}$ \\[6pt]
$\E[p_A p_B]$ & $(1,1)$ & $e^{\rho\sigma_A\sigma_B\, T}$ \\
$\E[p_A^2]$ & $(2,0)$ & $e^{\sigma_A^2 T}$ \\
$\E[p_B^2]$ & $(0,2)$ & $e^{\sigma_B^2 T}$ \\
$\E[p_A^{3/2}\,p_B^{1/2}]$ & $(\tfrac{3}{2},\tfrac{1}{2})$
    & $e^{(3\sigma_A^2 - \sigma_B^2
           + 6\rho\sigma_A\sigma_B)\,T/8}$ \\[4pt]
$\E[p_A^{1/2}\,p_B^{3/2}]$ & $(\tfrac{1}{2},\tfrac{3}{2})$
    & $e^{(-\sigma_A^2 + 3\sigma_B^2
            + 6\rho\sigma_A\sigma_B)\,T/8}$ \\
\bottomrule
\end{tabular}
\end{table}

\subsection{Expected return}\label{sec:expected}

Write $p_i = S_T^i / S_0^i$ for the price relatives.  The initial
equity is
\begin{equation}\label{eq:equity}
    \Pi_0 = C + (1-h)\,V_0,
\end{equation}
since the LP deposits $C$ as collateral and spends $(1-h)V_0$ to
purchase the unhedged fraction of LP tokens.  From
\eqref{eq:portfolio} the P\&L is
\begin{equation}\label{eq:pnl}
    \Pi_T - \Pi_0
    = \underbrace{V_0(\sqrt{p_A p_B} - 1)}_{\Delta V^{\text{LP}}}
      + R T
      + C r_f T
      - \frac{h V_0}{2}(p_A + p_B - 2)
      - \frac{h V_0}{2}(r_A + r_B)\,T,
\end{equation}
where $R$ is the reward rate in dollars per year (so that $RT$ is total
reward income over the horizon) and $r_f$ is the stablecoin supply
rate (per annum).  In Table~\ref{tab:calibration}, $R/V_0 = 0.54$
means annual reward income equals 54\% of LP value.

\begin{proposition}[Expected P\&L]\label{prop:expected_return}
Under zero-drift GBM,
\begin{equation}\label{eq:expected_pnl}
    \E[\Pi_T - \Pi_0]
    = V_0\bigl(e^{-\phi T} - 1\bigr)
      + R T + C r_f T
      - \frac{h V_0}{2}(r_A + r_B)\,T,
\end{equation}
where $\phi = (\sigma_A^2 + \sigma_B^2 - 2\rho\sigma_A\sigma_B)/8$
governs the expected rate of LP value decay due to price divergence.
\end{proposition}
\begin{proof}
Apply $\E[p_A] = \E[p_B] = 1$ (Table~\ref{tab:moments}) to
\eqref{eq:pnl}.  The hedge-related price term
$\E[\tfrac{h V_0}{2}(p_A + p_B - 2)] = 0$ vanishes, leaving only
the deterministic borrow cost $\tfrac{h V_0}{2}(r_A + r_B)T$.
\end{proof}

\begin{remark}
The expected LP value change $V_0(e^{-\phi T} - 1)$---the
structural cost of price divergence, commonly called impermanent
loss---is independent of~$h$.  The hedge ratio affects expected
return only through borrow costs, which are linear in~$h$.
The benefit of hedging is entirely through variance reduction:
the borrowing-based hedge targets price exposure, not the
non-linear divergence cost itself.
\end{remark}

\subsection{Variance of return}\label{sec:variance}

The only stochastic terms in \eqref{eq:pnl} are
$G := \sqrt{p_A p_B}$ (the LP value factor) and
$A := \tfrac{1}{2}(p_A + p_B)$ (the arithmetic mean of price
relatives).

\begin{proposition}[Variance decomposition]\label{prop:variance}
\begin{equation}\label{eq:var_decomp}
    \Var(\Pi_T) = V_0^2
        \bigl[v_{GG} + h^2\,v_{AA} - 2h\,v_{GA}\bigr],
\end{equation}
where
\begin{align}
    v_{GG} &:= \Var(G)
            = e^{\rho\sigma_A\sigma_B\,T} - e^{-2\phi T},
            \label{eq:vgg}\\
    v_{AA} &:= \Var(A)
            = \tfrac{1}{4}\bigl[
                e^{\sigma_A^2 T} + e^{\sigma_B^2 T}
                + 2\,e^{\rho\sigma_A\sigma_B\,T} - 4
            \bigr], \label{eq:vaa}\\
    v_{GA} &:= \Cov(G, A)
            = \tfrac{1}{2}\bigl[
                e^{(3\sigma_A^2 - \sigma_B^2
                    + 6\rho\sigma_A\sigma_B)\,T/8}
              + e^{(-\sigma_A^2 + 3\sigma_B^2
                    + 6\rho\sigma_A\sigma_B)\,T/8}
              - 2\,e^{-\phi T}
            \bigr]. \label{eq:vga}
\end{align}
\end{proposition}

\begin{proof}
From \eqref{eq:pnl}, $\Var(\Pi_T) = V_0^2\,\Var(G - hA)
= V_0^2[\Var(G) + h^2\Var(A) - 2h\Cov(G,A)]$.

For $v_{GG}$: $\Var(G) = \E[G^2] - (\E[G])^2
= e^{\rho\sigma_A\sigma_B T} - e^{-2\phi T}$
using Table~\ref{tab:moments} with $(a,b) = (1,1)$ and
$(\tfrac{1}{2},\tfrac{1}{2})$.

For $v_{AA}$:
$\Var(A) = \tfrac{1}{4}[\Var(p_A) + \Var(p_B) + 2\Cov(p_A,p_B)]$
with $\Var(p_i) = \E[p_i^2] - 1$ and
$\Cov(p_A,p_B) = \E[p_Ap_B] - 1$.

For $v_{GA}$:
$\Cov(G,A) = \tfrac{1}{2}[\E[Gp_A] + \E[Gp_B]] - \E[G]$
where $\E[Gp_A] = \E[p_A^{3/2}p_B^{1/2}]$ and
$\E[Gp_B] = \E[p_A^{1/2}p_B^{3/2}]$ from Table~\ref{tab:moments}.
\end{proof}

\begin{remark}\label{rem:quadratic}
The variance \eqref{eq:var_decomp} is a convex quadratic in $h$
with minimum at the \emph{minimum-variance hedge ratio}
$h_{\mathrm{mv}} = v_{GA}/v_{AA}$.
Under the calibration in Table~\ref{tab:calibration},
$h_{\mathrm{mv}} \approx 0.977$.
\end{remark}

\subsection{Unconstrained optimal hedge ratio}\label{sec:unconstrained}

Define the expected excess P\&L and its derivative:
\begin{align}
    \mu(h) &:= \E[\Pi_T - \Pi_0]
            = \mu_0 - c\,h, \label{eq:mu_h}\\
    c &:= \tfrac{V_0}{2}(r_A + r_B)\,T > 0, \label{eq:cost_c}
\end{align}
where $\mu_0 := V_0(e^{-\phi T}-1) + RT + Cr_fT$ collects the
$h$-independent terms.  Since $c > 0$, expected return decreases
linearly in $h$.

The Sharpe ratio is
\begin{equation}\label{eq:sharpe_h}
    \mathrm{SR}(h) = \frac{\mu_0 - ch}
        {V_0\sqrt{v_{GG} + h^2 v_{AA} - 2h\,v_{GA}}}.
\end{equation}

\begin{theorem}[Optimal hedge ratio]\label{thm:optimal_h}
Assume $\mu_0 > 0$ (the strategy is profitable at $h=0$) and
$c > 0$ (borrowing is costly).  The unconstrained maximizer of
$\mathrm{SR}(h)$ satisfies the linear first-order condition
\begin{equation}\label{eq:foc_linear}
    h\,(\mu_0\,v_{AA} - c\,v_{GA})
    = \mu_0\,v_{GA} - c\,v_{GG},
\end{equation}
with unique solution
\begin{equation}\label{eq:h_star}
    h^{*} = \frac{\mu_0\,v_{GA} - c\,v_{GG}}
                 {\mu_0\,v_{AA} - c\,v_{GA}}.
\end{equation}
When borrowing is free ($c = 0$), this reduces to the
minimum-variance hedge ratio $h_{\mathrm{mv}} = v_{GA}/v_{AA}$.
\end{theorem}

\begin{proof}
Setting $\dd\,\mathrm{SR}/\dd h = 0$ and clearing denominators
yields the condition
\begin{equation}\label{eq:foc}
    \mu'(h)\;\sigma^2(h) = \mu(h)\;\tfrac{1}{2}\,
        \frac{\dd\sigma^2}{\dd h}(h),
\end{equation}
where $\sigma^2(h) = V_0^2(v_{GG} + h^2 v_{AA} - 2hv_{GA})$.
Substituting $\mu'(h) = -c$ and
$\dd\sigma^2/\dd h = V_0^2(2hv_{AA} - 2v_{GA})$:
\[
    -c(v_{GG} + h^2 v_{AA} - 2hv_{GA})
    = (\mu_0 - ch)(hv_{AA} - v_{GA}).
\]
Expanding both sides, the $ch^2 v_{AA}$ terms cancel, leaving the
linear equation \eqref{eq:foc_linear}.  The second-order condition
$\mathrm{SR}''(h^{*}) < 0$ is verified in
Appendix~\ref{app:foc}.
\end{proof}

\begin{corollary}[Near-full unconstrained hedge]
\label{cor:near_full}
When $c/\mu_0 \ll 1$ (reward rates dominate borrow costs),
$h^{*} \approx h_{\mathrm{mv}} = v_{GA}/v_{AA}$.  For
token pairs with similar volatilities
($\sigma_A \approx \sigma_B$) and moderately positive
correlation, $h_{\mathrm{mv}}$ is close
to~1, so the unconstrained optimum calls for nearly complete
hedging.
\end{corollary}

\begin{proof}
From \eqref{eq:h_star}, as $c \to 0$,
$h^{*} \to v_{GA}/v_{AA} = h_{\mathrm{mv}}$.  Since
$v_{GA} \leq \sqrt{v_{GG}\,v_{AA}}$ by Cauchy--Schwarz,
$h_{\mathrm{mv}} \leq \sqrt{v_{GG}/v_{AA}}$.  When
$\sigma_A \approx \sigma_B$, $v_{GG} \approx v_{AA}$, so
$h_{\mathrm{mv}} \approx 1$.
\end{proof}

\begin{remark}\label{rem:interior_from_constraint}
\textbf{This is the paper's central result.}  Under the
calibration in Table~\ref{tab:calibration}, $h^{*} = 0.977$
(Theorem~\ref{thm:optimal_h}), meaning the unconstrained
optimum is near-full hedging.  However, at $h = 0.977$ the
liquidation probability exceeds 19\%
(Table~\ref{tab:liquidation_mc}), far above any reasonable
tolerance.  The \emph{liquidation constraint} forces the
feasible hedge ratio down to $h^{**} \approx 0.60$--$0.65$, where
liquidation probability drops to 1.4--2.3\%.  Thus the practical
interior optimum arises not from the Sharpe maximization
itself, but from the binding liquidation constraint, as
formalized in Section~\ref{sec:liquidation}.
\end{remark}

\subsection{Numerical illustration}

Using the calibration in Table~\ref{tab:calibration}
($\sigma_A = 0.922$, $\sigma_B = 1.084$, $\rho = 0.72$,
$r_A = 0.03$, $r_B = 0.15$, $R/V_0 = 0.54$, $C/V_0 = 2.0$,
$T = 0.25$), we compute:
\begin{align*}
    \phi &= 0.0732, \quad
    v_{GG} = 0.2331, \quad
    v_{AA} = 0.2431, \quad
    v_{GA} = 0.2376, \\
    \mu_0/V_0 &= 0.1369, \quad
    c/V_0 = 0.0225, \quad
    h_{\mathrm{mv}} = 0.977.
\end{align*}
The unconstrained solution \eqref{eq:h_star} gives
$h^{*} = 0.977 \approx 1$, confirming that the optimal
unconstrained hedge is nearly complete
(Corollary~\ref{cor:near_full}).  The analytical Sharpe ratio
rises steeply as $h \to h^{*}$:

\medskip
\begin{center}
\begin{tabular}{@{}ccccccccc@{}}
\toprule
$h$ & 0.0 & 0.3 & 0.5 & 0.6 & 0.7 & 0.8 & $h^{*}$ & 1.0 \\
\midrule
$\mathrm{SR}(h)$ & 0.28 & 0.39 & 0.53 & 0.65 & 0.87 & 1.29
    & 3.88 & 3.62 \\
\bottomrule
\end{tabular}
\end{center}

\medskip\noindent
However, the Monte Carlo simulation (which incorporates
liquidation losses as a 20\% collateral penalty) yields an
effective Sharpe peak at $h = 0.60$.  At $h^{*} = 0.977$ the
liquidation probability exceeds 19\%, destroying
risk-adjusted returns.  This confirms that the
\emph{liquidation constraint is the binding factor} that forces
$h^{**}$ well below $h^{*}$
(Remark~\ref{rem:interior_from_constraint}).

\section{Liquidation Probability via First Passage Time}
\label{sec:liquidation}

The liquidation constraint in \eqref{eq:objective} requires bounding
the probability that the LTV process exceeds $\ell_{\max}$ at any
point during the horizon $[0,T]$.  This is a first-passage-time
problem.

\subsection{LTV dynamics}

From \eqref{eq:ltv}, the LTV at time $t$ is
\begin{equation}\label{eq:ltv_process}
    \text{LTV}_t = \frac{h V_0}{2C}
        \left(\frac{S_t^A}{S_0^A} + \frac{S_t^B}{S_0^B}\right)
    = \frac{h V_0}{2C}(p_A(t) + p_B(t)),
\end{equation}
where $p_i(t) = S_t^i/S_0^i$ are GBM processes.
The initial value is $\text{LTV}_0 = hV_0/C$.

The process $Z_t := p_A(t) + p_B(t)$ is the sum of two correlated
geometric Brownian motions.  Exact first-passage-time distributions
for sums of GBMs are not available in closed form, so we employ a
moment-matching approximation.

\subsection{Moment-matched GBM approximation}

We approximate $Z_t$ by a single GBM $\hat{Z}_t$ with the same
first two moments.  Under zero-drift dynamics:
\begin{align}
    \E[Z_t] &= 2, \label{eq:z_mean}\\
    \E[Z_t^2] &= e^{\sigma_A^2 t} + e^{\sigma_B^2 t}
                + 2e^{\rho\sigma_A\sigma_B\,t}.
                \label{eq:z_second}
\end{align}
Matching $\hat{Z}_t = 2\exp((\tilde{\mu}-\tilde{\sigma}^2/2)t
+ \tilde{\sigma}\tilde{W}_t)$ to these moments gives the
approximation parameters:
\begin{align}
    \tilde{\mu} &= 0, \label{eq:tilde_mu}\\
    \tilde{\sigma}^2 &= \frac{1}{t}\ln\!\left(
        \frac{e^{\sigma_A^2 t} + e^{\sigma_B^2 t}
              + 2e^{\rho\sigma_A\sigma_B\,t}}{4}
    \right). \label{eq:tilde_sigma}
\end{align}
For small $t$, $\tilde{\sigma}^2 \approx
(\sigma_A^2 + \sigma_B^2 + 2\rho\sigma_A\sigma_B)/4$, which is the
variance of the equally-weighted portfolio.

\subsection{First passage time to liquidation}

Define the first passage time
$\tau = \inf\{t \geq 0 : \text{LTV}_t \geq \ell_{\max}\}$.
Under the GBM approximation, $\text{LTV}_t/\text{LTV}_0
\approx \hat{Z}_t/2$, and the liquidation boundary in log-space
is $b = \ln(\ell_{\max}/\text{LTV}_0)$.

\begin{proposition}[Approximate liquidation probability]
\label{prop:liquidation}
Under the moment-matched GBM approximation with drift $\tilde{\mu}=0$
and volatility $\tilde{\sigma}$, define the log-process drift
$\nu := \tilde{\mu} - \tilde{\sigma}^2/2 = -\tilde{\sigma}^2/2$.
The probability of liquidation during $[0,T]$ is
\begin{equation}\label{eq:liquidation_prob}
    \Pr(\tau \leq T) \approx \Phi\!\left(
        \frac{-b - \frac{\tilde{\sigma}^2}{2}\,T}
             {\tilde{\sigma}\sqrt{T}}
    \right)
    + \frac{\text{LTV}_0}{\ell_{\max}}
    \,\Phi\!\left(
        \frac{-b + \frac{\tilde{\sigma}^2}{2}\,T}
             {\tilde{\sigma}\sqrt{T}}
    \right),
\end{equation}
where $b = \ln(\ell_{\max} C / (hV_0)) = \ln(\ell_{\max}/\text{LTV}_0)$
is the log-distance to the liquidation barrier,
$\Phi$ is the standard normal CDF, and the coefficient
$e^{2\nu b/\tilde{\sigma}^2} = e^{-b} = \text{LTV}_0/\ell_{\max}$
follows from substituting $\nu = -\tilde{\sigma}^2/2$.
\end{proposition}

\noindent
Note that $\tilde{\sigma}$ as defined in
\eqref{eq:tilde_sigma} depends on the horizon~$t$; for the
calibrated parameters it varies by less than 2\% over
$t \in [30, 180]$~days, so we treat it as approximately constant.

\begin{proof}
Under the GBM approximation, $\ln(\hat{Z}_t/\hat{Z}_0)$ is a
Brownian motion with drift $\nu = \tilde{\mu} - \tilde{\sigma}^2/2
= -\tilde{\sigma}^2/2$ and diffusion~$\tilde{\sigma}$.
The first-passage-time probability for a drifted Brownian motion
to a fixed barrier is a classical result based on the reflection
principle \citep{karatzas1991}, yielding
\eqref{eq:liquidation_prob} with drift parameter~$\nu$.
\end{proof}

\noindent
The moment-matching approximation captures the first two moments
of $Z_t$ but not higher-order properties such as skewness, and
the first-passage-time formula~\eqref{eq:liquidation_prob}
assumes a constant diffusion coefficient whereas $\tilde{\sigma}$
in~\eqref{eq:tilde_sigma} is weakly $t$-dependent.
Table~\ref{tab:analytical_vs_mc} validates the approximation
against Monte Carlo for the calibrated parameters, confirming
errors below 0.2\,pp for $h \leq 0.70$.

\subsection{Constrained optimal hedge ratio}

The liquidation probability \eqref{eq:liquidation_prob} is
increasing in $h$ (since $b$ decreases as $h$ increases).  Define
\begin{equation}\label{eq:h_max}
    \bar{h}(\alpha) := \sup\{h \in [0,1] :
        \Pr(\tau \leq T) \leq \alpha\},
\end{equation}
the maximum feasible hedge ratio for a given liquidation tolerance
$\alpha$.  The constrained optimal hedge ratio is then
\begin{equation}\label{eq:h_constrained}
    h^{**} = \min(h^{*},\; \bar{h}(\alpha)),
\end{equation}
where $h^{*}$ is the unconstrained optimum from
Theorem~\ref{thm:optimal_h}.

\begin{remark}
Under the calibration in Table~\ref{tab:calibration}, the
unconstrained optimum is $h^{*} = 0.977$
(Theorem~\ref{thm:optimal_h}), at which the liquidation
probability exceeds 19\% (Table~\ref{tab:liquidation_mc}).
For any reasonable tolerance $\alpha \leq 5\%$, the constraint
forces $\bar{h}(\alpha) < h^{*}$, yielding an interior solution
$h^{**} \approx 0.60$--$0.65$ where the Sharpe ratio (inclusive of
liquidation losses) is maximized.
\end{remark}

\section{Monte Carlo Simulation}\label{sec:simulation}

\subsection{Calibration}

We calibrate the model to on-chain data from a constant-product
AMM pool and a lending protocol on the Sui blockchain.
Token~$A$ is SUI and token~$B$ is NS (NaviSwap governance token).
Volatilities and correlation are estimated from 91~days of daily
CoinGecko price data using log-return standard deviations and the
Pearson correlation of log returns.

\begin{table}[H]
\centering
\caption{Baseline calibration parameters.}
\label{tab:calibration}
\begin{tabular}{@{}lll@{}}
\toprule
Parameter & Symbol & Value \\
\midrule
Token $A$ annualized volatility & $\sigma_A$ & 0.922 \\
Token $B$ annualized volatility & $\sigma_B$ & 1.084 \\
Correlation & $\rho$ & 0.72 \\
Token $A$ borrow rate & $r_A$ & 0.03 \\
Token $B$ borrow rate & $r_B$ & 0.15 \\
LP reward rate (farm + fees) & $R/V_0$ & 0.54 \\
Stablecoin supply rate & $r_f$ & 0.04 \\
Max LTV & $\ell_{\max}$ & 0.80 \\
Collateral / LP ratio & $C/V_0$ & 2.0 \\
Horizon & $T$ & 0.25 yr (90~days) \\
\bottomrule
\end{tabular}
\end{table}

\noindent
The LP reward rate of 54\% includes both farm incentives
(${\approx}\,48\%$) and trading fees (${\approx}\,6\%$) observed on
the pool at the time of calibration.  The asymmetry in borrow rates
($r_B \gg r_A$) reflects the lower liquidity and higher demand for
borrowing token~$B$.

These reward rates reflect the token emission incentives commonly
used by DeFi protocols to attract liquidity during their growth
phase.  Unlike yields in traditional finance, which are bounded by
the risk-free rate and credit spreads, DeFi farming rewards are
funded by protocol-native token emissions and can exceed 50\% APR
for early-stage pools.  However, such rates tend to decline over
time as total value locked (TVL) increases and emission schedules
taper.  To ensure that our conclusions are not contingent on
elevated reward rates, Table~\ref{tab:sensitivity_apr} reports
the optimal hedge ratio across a wide range of reward rates
($R/V_0$ from 10\% to 100\%).  The key finding---that the optimal
hedge ratio lies in the range
$h^{**} \in [0.50,\, 0.70]$---is robust for $R/V_0 \geq 30\%$,
and the strategy remains viable (positive Sharpe ratio) for
$R/V_0 \geq 20\%$.

\subsection{Simulation design}

We simulate $N = 30{,}000$ paths of the correlated GBM pair
$(S_t^A, S_t^B)$ with daily time steps ($\Delta t = 1/365$) over
$T = 90$~days.  Correlated normal increments are generated via the
Cholesky decomposition $Z_2 = \rho Z_1 + \sqrt{1-\rho^2}\,Z_2'$,
where $Z_1, Z_2'$ are independent standard normals.

For each path and each hedge ratio
$h \in \{0, 0.20, 0.40, 0.50, 0.60, 0.65, 0.70, 0.80, 1.0\}$, we
compute:
\begin{enumerate}
    \item The daily LTV.  If $\text{LTV}_t \geq \ell_{\max}$ for
          any~$t$, the path is flagged as liquidated and assigned a
          penalty loss of $0.20 \times C$.  This figure aggregates
          the protocol liquidation bonus (typically 5--10\% of the
          repaid debt), market-impact slippage during forced
          selling, and the opportunity cost of the unwound position.
          On Sui-based protocols such as NAVI, the liquidation
          bonus is 5\%; the 20\% total accounts for adverse
          slippage in the thin markets characteristic of
          long-tail tokens.
    \item LP rewards, claimed every 14~days and applied as partial
          debt repayment.
    \item The terminal P\&L $\Pi_T - \Pi_0$ for surviving paths.
\end{enumerate}
Let $\bar{r}$ denote the sample mean ROE over the $T$-day period
and $\hat{\sigma}$ its sample standard deviation.  The annualized
Sharpe ratio is
\begin{equation}\label{eq:sr_annualized}
    \mathrm{SR} = \frac{\bar{r} - r_f\,(T/365)}{\hat{\sigma}}
        \cdot \sqrt{\frac{365}{T}},
\end{equation}
where $T$ is the horizon in days (here $T = 90$; note that the
analytical sections use $T$ in years), scaling the per-period excess
return and volatility to an annual basis.

\section{Results}\label{sec:results}

\subsection{Optimal hedge ratio}\label{sec:results_optimal}

Table~\ref{tab:results_hedge} reports the simulation results across
hedge ratios.  The Sharpe ratio peaks at $h \in [0.60,\, 0.65]$
(raw SR $= 0.93$--$0.95$), with $h = 0.65$ marginally higher but $h = 0.60$ within the
flat region of the optimum.

\begin{table}[H]
\centering
\caption{Monte Carlo results by hedge ratio ($N = 30{,}000$ paths,
$T = 90$~days).  ``Raw'' Sharpe excludes transaction costs;
``+tx'' includes a one-time 0.3\% borrow fee and gas costs.}
\label{tab:results_hedge}
\begin{tabular}{@{}rrrrrrrr@{}}
\toprule
$h$ (\%) & E[ROE] & Std & SR (raw) & SR (+tx) & P(loss) & P(liq) & 5\% VaR \\
\midrule
0  & $+3.58$ & 16.0 & 0.453 & 0.326 & 49.4\% &  0.0\% & $-15.4$ \\
20 & $+3.71$ & 13.6 & 0.553 & 0.401 & 46.9\% &  0.0\% & $-12.6$ \\
40 & $+3.78$ & 10.8 & 0.720 & 0.523 & 42.9\% &  0.1\% &  $-9.2$ \\
50 & $+3.69$ &  9.1 & 0.836 & 0.599 & 39.9\% &  0.4\% &  $-7.5$ \\
\textbf{60} & $\mathbf{+3.32}$ & \textbf{7.4} & \textbf{0.931}
    & \textbf{0.636} & \textbf{36.0\%} & \textbf{1.4\%}
    & $\mathbf{-5.8}$ \\
\textbf{65} & $\mathbf{+3.06}$ & \textbf{6.7} & \textbf{0.951}
    & \textbf{0.622} & \textbf{32.6\%} & \textbf{2.3\%}
    & $\mathbf{-5.0}$ \\
70 & $+2.76$ &  6.3 & 0.914 & 0.565 & 29.6\% &  3.4\% &  $-4.2$ \\
80 & $+1.96$ &  6.6 & 0.640 & 0.299 & 17.8\% &  7.0\% & $-19.3$ \\
100 & $-0.32$ & 10.2 & $-0.030$ & $-0.258$ & 19.9\% & 19.2\%
    & $-21.2$ \\
\bottomrule
\end{tabular}
\smallskip

\noindent\small E[ROE], Std, and VaR are in percentage points.
E[ROE] and Std include transaction costs (matching the SR\,+\,tx column).
\end{table}

Figure~\ref{fig:sharpe_vs_h} visualizes the trade-off: the Sharpe
ratio peaks at $h = 0.60$ and the liquidation probability rises
sharply beyond $h = 0.70$.

\begin{figure}[H]
\centering
\includegraphics[width=0.85\textwidth]{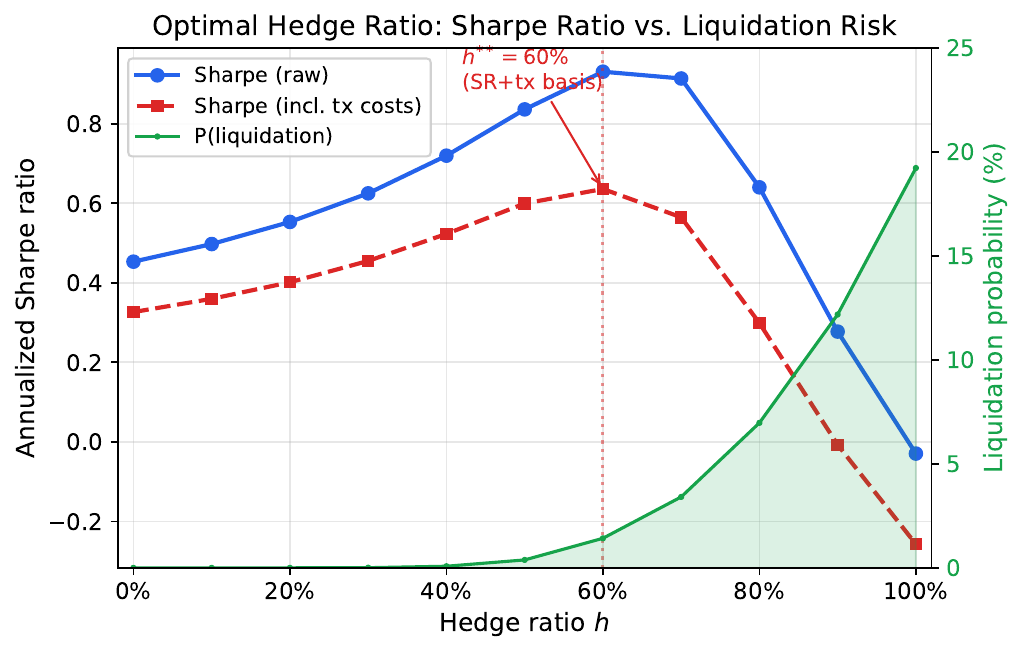}
\caption{Sharpe ratio (left axis) and liquidation probability
(right axis, green) as functions of the hedge ratio~$h$.  The
optimal $h^{**} = 60\%$ (on the SR\,+\,tx basis; raw SR is
marginally higher at $h = 0.65$) balances risk-adjusted return
against liquidation risk.}
\label{fig:sharpe_vs_h}
\end{figure}

Several patterns emerge:
\begin{itemize}
    \item \textbf{Interior optimum.}  The Sharpe ratio peaks in
          the range $h = 0.60$--$0.65$ (raw: 0.93--0.95) and declines
          on both sides, confirming the analytical result.
    \item \textbf{Risk--return trade-off.}  Increasing $h$ from 0 to
          0.60 reduces standard deviation from 16.0\% to 7.4\%
          while sacrificing only 0.26pp of expected return.  Beyond
          $h = 0.70$, liquidation risk rises sharply (3.4\% at
          $h = 0.70$, 19.2\% at $h = 1.0$), destroying
          risk-adjusted returns.
    \item \textbf{Transaction costs.}  The one-time borrow fee
          (0.3\% of borrowed amount) reduces the Sharpe ratio by
          0.13--0.35 depending on hedge ratio, but does not shift
          the optimal $h$.
\end{itemize}

\subsection{Liquidation analysis}

Table~\ref{tab:liquidation_mc} shows the liquidation statistics for
the fully hedged case ($h = 1.0$, initial LTV = 50\%).

\begin{table}[H]
\centering
\caption{Liquidation statistics at $h = 1.0$ ($N = 30{,}000$,
$T = 90$~days).}
\label{tab:liquidation_mc}
\begin{tabular}{@{}lrr@{}}
\toprule
Metric & No claims & Claim/14d \\
\midrule
Liquidation probability & 23.0\% & 19.2\% \\
Mean max LTV & 70.2\% & 67.8\% \\
95th pctl max LTV & 113.0\% & 108.2\% \\
99th pctl max LTV & 150.0\% & 144.5\% \\
\bottomrule
\end{tabular}
\end{table}

\noindent
Regular reward claims reduce liquidation probability by
${\approx}\,4$pp.  Liquidated paths have an average price multiple
of $1.63\times$ (i.e., both tokens approximately doubled), while
safe paths average $0.82\times$.  This confirms that liquidation is
driven by \emph{rising} prices increasing debt value against fixed
stablecoin collateral.

At the optimal $h^{**} = 0.60$, initial LTV drops to 30\% and
liquidation probability falls to 1.4\%---well within the 5\%
tolerance assumed in Section~\ref{sec:liquidation}.

\paragraph{Accuracy of the analytical approximation.}
Table~\ref{tab:analytical_vs_mc} compares the first-passage-time
approximation \eqref{eq:liquidation_prob} with the Monte Carlo
liquidation probability (without reward claims, to isolate the
approximation error).

\begin{table}[H]
\centering
\caption{Analytical (Eq.~\ref{eq:liquidation_prob}) vs.\ Monte Carlo
liquidation probability ($N = 50{,}000$ for tighter confidence
intervals; main results use $N = 30{,}000$).  $T = 90$~days, SUI/NS.
MC(no) excludes reward claims; MC(cl) includes biweekly claims.}
\label{tab:analytical_vs_mc}
\begin{tabular}{@{}rrrrrr@{}}
\toprule
$h$ (\%) & $\text{LTV}_0$ & $b$ & Analytical & MC (no claims) & MC (claims) \\
\midrule
30 & 15.0\% & 1.674 & 0.01\% &  0.01\% &  0.01\% \\
40 & 20.0\% & 1.386 & 0.13\% &  0.12\% &  0.07\% \\
50 & 25.0\% & 1.163 & 0.66\% &  0.66\% &  0.39\% \\
60 & 30.0\% & 0.981 & 2.05\% &  2.02\% &  1.41\% \\
70 & 35.0\% & 0.827 & 4.82\% &  4.65\% &  3.34\% \\
80 & 40.0\% & 0.693 & 9.34\% &  8.97\% &  6.83\% \\
100 & 50.0\% & 0.470 & 24.18\% & 22.88\% & 19.08\% \\
\bottomrule
\end{tabular}
\end{table}

\noindent
The analytical approximation closely matches MC without reward claims
(error $< 0.2$pp for $h \leq 70\%$ and $+1.3$pp at $h = 100\%$),
validating the moment-matched GBM approximation of
Section~\ref{sec:liquidation}.  The approximation slightly
\emph{overestimates} liquidation probability, making it a
conservative bound.  The gap between MC(no) and MC(cl) quantifies
the risk-reduction benefit of biweekly reward claims.

Figure~\ref{fig:return_risk} decomposes the return--risk trade-off,
showing that expected return decreases slowly with $h$ while
standard deviation drops sharply up to $h \approx 0.70$.

\begin{figure}[H]
\centering
\includegraphics[width=0.85\textwidth]{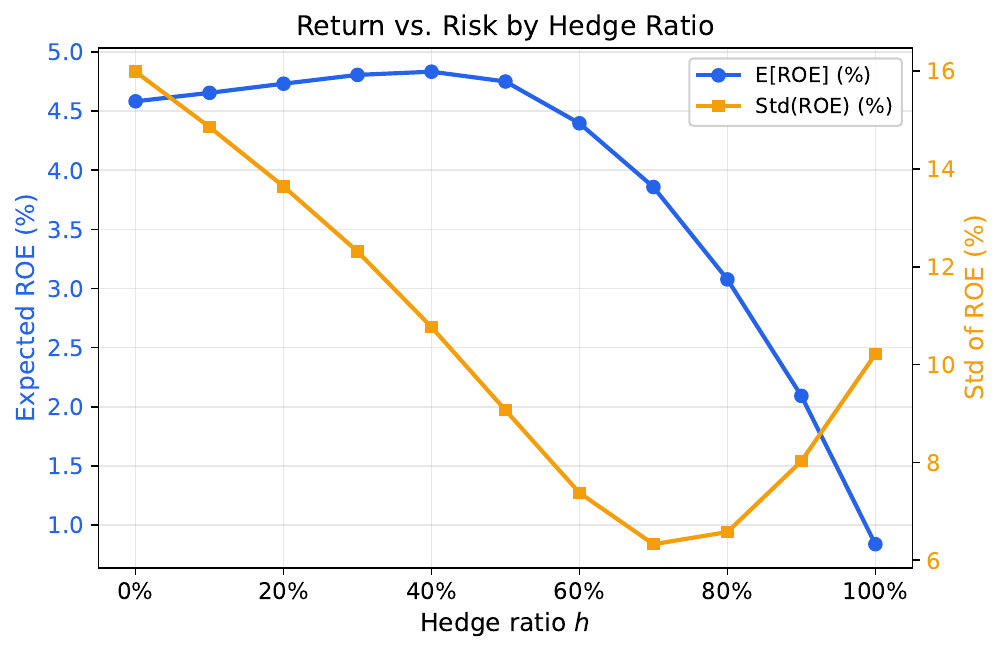}
\caption{Expected ROE (left axis) and standard deviation of ROE
(right axis) as functions of the hedge ratio.  The variance
reduction from hedging is substantial up to $h \approx 0.70$,
after which liquidation losses cause both return and risk to
deteriorate.}
\label{fig:return_risk}
\end{figure}

\subsection{Sensitivity analysis}

We examine how the optimal hedge ratio varies with the correlation
$\rho$ and the borrow rate $r_B$.

Figure~\ref{fig:sensitivity_rho} shows that the optimal $h^{**}$
is robust to changes in correlation: across all tested values
($\rho \in [0, 0.9]$), the Sharpe-maximizing hedge ratio remains
in the range 60--70\%.  Higher correlation improves the overall
Sharpe level (less price divergence) but does not materially shift the optimum.

\begin{figure}[H]
\centering
\includegraphics[width=0.85\textwidth]{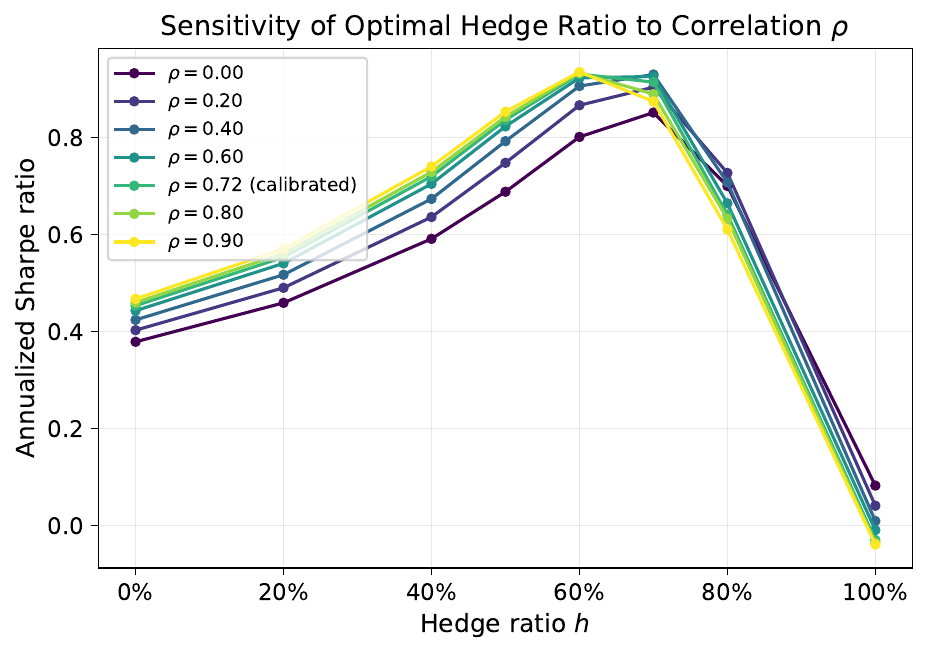}
\caption{Sharpe ratio vs.\ hedge ratio for different correlation
values~$\rho$.  The optimal $h^{**}$ remains between 60\% and
70\% regardless of correlation.}
\label{fig:sensitivity_rho}
\end{figure}

Figure~\ref{fig:sensitivity_borrow} shows sensitivity to the
token~$B$ borrow rate.  Higher borrow costs shift the optimal
$h^{**}$ slightly downward (from 70\% at $r_B = 5\%$ to 60\%
at $r_B = 30\%$), consistent with Proposition~\ref{prop:expected_return}:
borrow cost is the only channel through which $h$ affects expected
return.

\begin{figure}[H]
\centering
\includegraphics[width=0.85\textwidth]{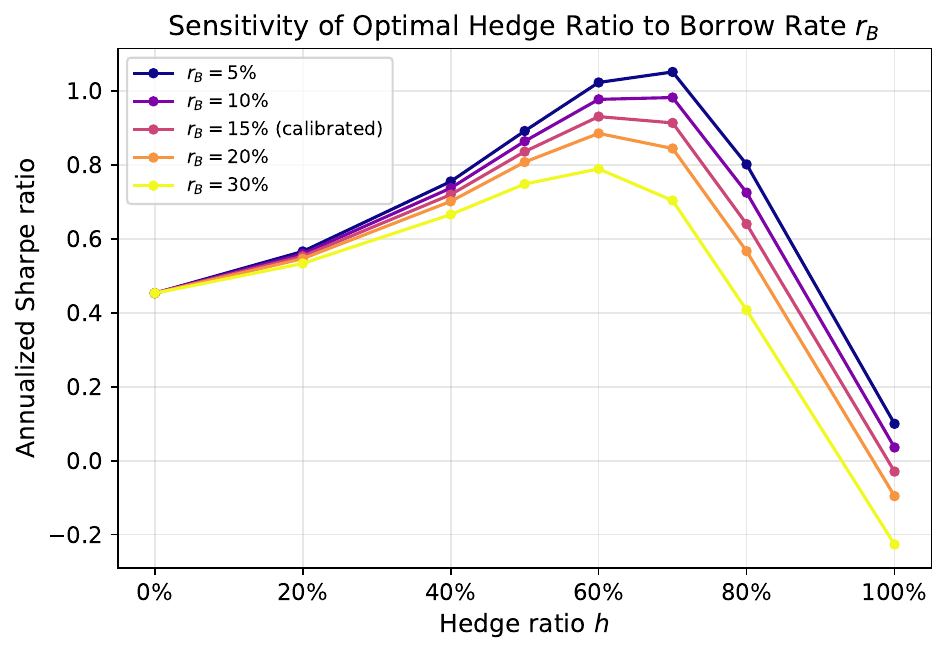}
\caption{Sharpe ratio vs.\ hedge ratio for different borrow rates
$r_B$.  Higher borrow costs reduce overall Sharpe and shift the
optimal hedge ratio slightly downward.}
\label{fig:sensitivity_borrow}
\end{figure}

\paragraph{Sensitivity to LP reward rate.}
Table~\ref{tab:sensitivity_apr} shows how the optimal hedge ratio
varies with the LP reward rate $R/V_0$.  The reward rate is the
primary driver of strategy viability: at $R/V_0 = 10\%$, the
strategy is unprofitable at any hedge ratio, while at
$R/V_0 \geq 40\%$, hedging becomes beneficial.  The optimal
hedge ratio increases with rewards (from $h^{**} = 40\%$ at
$R/V_0 = 20\%$ to $h^{**} = 70\%$ at $R/V_0 = 70\%$), because
higher rewards offset borrow costs and increase the margin for
absorbing liquidation losses.

\begin{table}[H]
\centering
\caption{Optimal hedge ratio by LP reward rate
($N = 30{,}000$, $T = 90$~days, SUI/NS parameters).}
\label{tab:sensitivity_apr}
\begin{tabular}{@{}rrrl@{}}
\toprule
$R/V_0$ & $h^{**}$ & SR & Remark \\
\midrule
10\% & 0\%  & $-0.00$ & Strategy unprofitable \\
20\% & 40\% &  0.12  & Marginal viability \\
30\% & 50\% &  0.30  & \\
40\% & 60\% &  0.53  & \\
54\% & 60\% &  0.93  & Calibrated value \\
70\% & 70\% &  1.45  & \\
100\% & 70\% &  2.41 & \\
\bottomrule
\end{tabular}
\smallskip

\noindent\small $h^{**}$ is evaluated on the grid of
Table~\ref{tab:results_hedge}, which includes $h = 0.65$.
At $R/V_0 = 54\%$, the raw Sharpe ratio peaks at $h = 0.65$
(SR\,$= 0.95$); the transaction-cost-adjusted Sharpe peaks at
$h = 0.60$ (Table~\ref{tab:results_hedge}).
Table~\ref{tab:robustness_pairs} reports $h^{**} = 65\%$ for
SUI/NS on a finer 5\% grid, consistent with the raw-SR ranking.
\end{table}

\subsection{Robustness across token pairs}

To verify that the optimal hedge ratio is not an artifact of the
SUI/NS calibration, we replicate the Monte Carlo analysis for three
additional token pairs on low-gas blockchains where delta-neutral LP
strategies are practically viable.  Table~\ref{tab:robustness_pairs}
reports the calibration parameters, estimated from 365~days of daily
CoinGecko price data, together with representative lending and LP
reward rates from each chain's dominant protocols.

\begin{table}[H]
\centering
\caption{Calibration parameters and optimal hedge ratios for
additional token pairs ($N = 30{,}000$, $T = 90$~days, $C/V_0 = 2.0$,
5\% hedge-ratio grid).}
\label{tab:robustness_pairs}
\begin{tabular}{@{}llrrrrrrrr@{}}
\toprule
Pair & Chain & $\sigma_A$ & $\sigma_B$ & $\rho$ & $r_A$ & $r_B$ & LP APR & $h^{**}$ & SR \\
\midrule
SUI/NS  & Sui      & 92\% & 108\% & 0.72 & 3\% & 15\% & 54\% & 65\% & 0.95 \\
SOL/RAY & Solana   & 80\% & 110\% & 0.83 & 5\% &  8\% & 40\% & 60\% & 0.68 \\
SOL/JUP & Solana   & 80\% & 100\% & 0.86 & 5\% &  6\% & 35\% & 65\% & 0.67 \\
ETH/ARB & Arbitrum & 74\% & 103\% & 0.82 & 3\% &  5\% & 30\% & 60\% & 0.54 \\
\bottomrule
\end{tabular}
\end{table}

\noindent
Using a fine 5\% grid (rather than the coarser grid in
Table~\ref{tab:results_hedge}), the Sharpe-maximizing hedge ratio
falls in the range $h^{**} \in [60\%,\, 65\%]$ across all four
pairs---spanning three blockchains, correlations from 0.72 to 0.86,
and annual volatilities from 74\% to 110\%.  The absolute Sharpe
ratio varies across pairs (0.54 for ETH/ARB to 0.95 for SUI/NS),
reflecting differences in LP reward rates, but the optimal hedge
ratio is tightly clustered.  This
confirms that the interior optimum arises from the \emph{structural}
interaction between hedging benefit and liquidation risk, not from
pair-specific parameter values.

\paragraph{Sensitivity to volatility level.}
The baseline calibration uses a 91-day estimation window, and
crypto-asset volatilities can shift substantially across
regimes.  Table~\ref{tab:sensitivity_vol} reports the optimal
hedge ratio when both volatilities are scaled by $\pm 20\%$.

\begin{table}[H]
\centering
\caption{Sensitivity to volatility level ($N = 30{,}000$,
$T = 90$~days, all other parameters at baseline).}
\label{tab:sensitivity_vol}
\begin{tabular}{@{}rrrrrr@{}}
\toprule
$\sigma$ scaling & $\sigma_A$ & $\sigma_B$ & $h^{**}$ & SR & P(liq) \\
\midrule
$-20\%$ & 74\% &  87\% & 70\% & 1.64 & 1.1\% \\
Baseline & 92\% & 108\% & 65\% & 0.95 & 2.3\% \\
$+20\%$ & 111\% & 130\% & 50\% & 0.55 & 1.4\% \\
\bottomrule
\end{tabular}
\end{table}

\noindent
The optimal hedge ratio shifts from 50\% (high-vol regime) to
70\% (low-vol regime), a 20\,pp range centered on the baseline
optimum.  This is intuitive: lower volatility reduces the
probability of large price moves that breach the LTV threshold,
allowing higher hedge ratios.  The finding reinforces the
practical guideline to hedge conservatively (50--60\%) during
high-volatility periods and more aggressively (65--70\%) during
calm markets.

\paragraph{Sensitivity to liquidation penalty.}
The baseline simulation assumes a 20\% collateral penalty upon
liquidation.  Since actual penalties vary across protocols
(5--10\% liquidation bonus plus variable slippage),
Table~\ref{tab:sensitivity_penalty} examines the impact of
alternative penalty assumptions.

\begin{table}[H]
\centering
\caption{Sensitivity to liquidation penalty ($N = 30{,}000$,
$T = 90$~days, SUI/NS parameters).}
\label{tab:sensitivity_penalty}
\begin{tabular}{@{}rrrr@{}}
\toprule
Penalty & $h^{**}$ & SR & P(liq) at $h^{**}$ \\
\midrule
10\% & 80\% & 1.18 & 7.0\% \\
20\% & 65\% & 0.95 & 2.3\% \\
30\% & 60\% & 0.85 & 1.4\% \\
\bottomrule
\end{tabular}
\end{table}

\noindent
The optimal hedge ratio is sensitive to the liquidation penalty:
at 10\% penalty, the optimizer tolerates 7\% liquidation
probability and pushes $h^{**}$ to 80\%; at 30\%, it retreats to
60\% with only 1.4\% liquidation probability.  This confirms
that the liquidation penalty is a key model input, and
practitioners should calibrate it to their specific protocol's
liquidation bonus, typical slippage, and market depth.

\subsection{Sensitivity to collateral ratio}

The baseline analysis fixes $C/V_0 = 2.0$.  Since the collateral
ratio directly determines the liquidation buffer, we examine how
$h^{**}$ varies with $C/V_0$.
Table~\ref{tab:cv_sensitivity} reports the optimal hedge ratio for
$C/V_0 \in [1.2, 5.0]$ using the SUI/NS calibration.

\begin{table}[H]
\centering
\caption{Optimal hedge ratio and initial LTV as functions of
the collateral ratio $C/V_0$ ($N = 30{,}000$, $T = 90$~days,
SUI/NS parameters).}
\label{tab:cv_sensitivity}
\begin{tabular}{@{}rrrrrr@{}}
\toprule
$C/V_0$ & $h^{**}$ & SR & P(liq) & E[ROE] & Init LTV \\
\midrule
1.2 & 30\% & 0.62 & 0.3\% & $+6.15$\% & 25.0\% \\
1.5 & 50\% & 0.71 & 2.3\% & $+4.64$\% & 33.3\% \\
1.8 & 60\% & 0.84 & 2.5\% & $+4.27$\% & 33.3\% \\
2.0 & 65\% & 0.95 & 2.3\% & $+4.16$\% & 32.5\% \\
2.5 & 80\% & 1.27 & 2.2\% & $+3.80$\% & 32.0\% \\
3.0 & 80\% & 1.68 & 0.7\% & $+3.74$\% & 26.7\% \\
4.0 & 90\% & 3.00 & 0.2\% & $+3.24$\% & 22.5\% \\
5.0 & 100\% & 3.97 & 0.1\% & $+2.83$\% & 20.0\% \\
\bottomrule
\end{tabular}
\end{table}

\noindent
Two patterns emerge.  First, $h^{**}$ increases monotonically with
$C/V_0$: more collateral relaxes the liquidation constraint, allowing
higher hedge ratios.  At $C/V_0 = 5.0$, the constraint becomes
non-binding and $h^{**}$ approaches the unconstrained optimum
$h^{*} \approx 0.977$.  Second, the initial LTV at the optimum
clusters around 25--33\%, suggesting that \textbf{maintaining an
initial LTV near 30\%} is a natural anchor across collateral
configurations.

However, a higher $C/V_0$ also dilutes capital efficiency: E[ROE]
declines from $+6.15$\% at $C/V_0 = 1.2$ to $+2.83$\% at
$C/V_0 = 5.0$, because a larger share of capital is locked in
low-yielding stablecoin collateral.  The elevated Sharpe ratios at
high $C/V_0$ reflect reduced variance rather than higher returns.
In practice, LPs face a capital allocation trade-off between risk
reduction and return, and the choice of $C/V_0$ should reflect
individual risk preferences and opportunity costs.

\subsection{Rebalancing frequency}

We compare five rebalancing strategies, all starting at
$h = 0.60$, in Table~\ref{tab:rebalancing}.  The ``threshold''
strategy rebalances when either token's effective hedge ratio
drifts more than the specified number of percentage points from
the target.

\begin{table}[H]
\centering
\caption{Rebalancing strategy comparison ($h = 0.60$,
$N = 30{,}000$, $T = 90$~days).}
\label{tab:rebalancing}
\begin{tabular}{@{}lrrrrrr@{}}
\toprule
Strategy & E[ROE] & Std & SR & P(liq) & Avg rebal. & Cost \\
\midrule
No rebalance     & $+4.41$ & 7.15 & 0.965 & 1.4\% & 0.0 & 0\% \\
Threshold 20pp   & $+4.98$ & 7.25 & 1.108 & 1.3\% & 0.2 & $<0.01$\% \\
Threshold 15pp   & $+5.14$ & 7.28 & 1.148 & 1.3\% & 0.5 & $<0.01$\% \\
Threshold 10pp   & $+5.26$ & 7.32 & 1.178 & 1.2\% & 1.2 & $<0.01$\% \\
Every 14 days    & $+5.30$ & 7.32 & 1.186 & 1.2\% & 6.0 & $<0.01$\% \\
Every 30 days    & $+5.36$ & 7.34 & 1.198 & 1.2\% & 3.0 & $<0.01$\% \\
\bottomrule
\end{tabular}
\smallskip

\noindent\small E[ROE] and Std are in percentage points.
E[ROE] here excludes the one-time entry/exit costs (borrow fee and gas)
reported in Table~\ref{tab:results_hedge}; the difference of
approximately 1.1 percentage points accounts for the apparent
discrepancy between the two tables.
\end{table}

The key finding is that \textbf{any} rebalancing improves
risk-adjusted returns over the static hedge.  The threshold-based
strategies achieve most of the benefit (Sharpe 1.15 at 15pp
threshold) with far fewer trades than periodic rebalancing.
Rebalancing costs are negligible on Sui (gas ${\approx}\,$\$0.01
per transaction).

A 15pp drift threshold represents a practically appealing rule:
with the calibrated volatilities, the median time to trigger is
approximately 62--85~days, meaning the LP rebalances roughly once
per quarter.

\subsection{Robustness to jump risk}\label{sec:jump}

The baseline model assumes geometric Brownian motion, which
understates the frequency of large price moves observed in
cryptocurrency markets.  To assess whether the optimal hedge ratio
is an artifact of the continuous-path assumption, we repeat the
Monte Carlo analysis under a Merton jump-diffusion model
\citep{merton1976}.

Each token price follows
\begin{equation}\label{eq:jump_diffusion}
    \frac{\dd S_t^i}{S_t^i}
    = \bigl(\mu_i - \lambda\kappa\bigr)\,\dd t
      + \sigma_i^{\text{d}}\,\dd W_t^i
      + J_i\,\dd N_t,
\end{equation}
where $N_t$ is a Poisson process with intensity $\lambda$,
$J_i = e^{Z_i} - 1$ with
$Z_i \sim \mathcal{N}(\mu_J,\, \sigma_J^2)$ is the random jump
size, and $\kappa = \E[J_i] = e^{\mu_J + \sigma_J^2/2} - 1$ is
the compensator ensuring $\E[\dd S/S] = \mu_i\,\dd t$ in the
absence of diffusion.  Let $N_t^c$, $N_t^{A}$, $N_t^{B}$ be
independent Poisson processes with intensities $\lambda\rho_J$,
$\lambda(1-\rho_J)$, $\lambda(1-\rho_J)$ respectively.
Token~$i$ experiences jumps from both $N_t^c$ (common shocks
affecting both tokens simultaneously) and $N_t^{i}$
(idiosyncratic shocks); the total jump intensity per token
is~$\lambda$.

We calibrate the jump parameters to stylized facts of crypto
markets: $\lambda = 4$~jumps/year, $\mu_J = -0.05$ (slight
downward bias), $\sigma_J = 0.15$, and $\rho_J = 0.80$ (high
co-jump probability).  The diffusion volatilities
$\sigma_i^{\text{d}}$ are reduced so that the \emph{total}
annualized variance matches the GBM calibration:
$(\sigma_i^{\text{d}})^2 = \sigma_i^2
  - \lambda(\mu_J^2 + \sigma_J^2)$.

\begin{table}[H]
\centering
\caption{GBM vs.\ Merton jump-diffusion ($N = 30{,}000$,
$T = 90$~days, SUI/NS parameters).  Jump parameters:
$\lambda = 4$/yr, $\mu_J = -0.05$, $\sigma_J = 0.15$,
$\rho_J = 0.80$.}
\label{tab:jump_comparison}
\begin{tabular}{@{}rrrrrrrr@{}}
\toprule
 & \multicolumn{3}{c}{GBM} & & \multicolumn{3}{c}{Jump-diffusion} \\
\cmidrule{2-4} \cmidrule{6-8}
$h$ (\%) & SR & P(liq) & 5\% VaR & & SR & P(liq) & 5\% VaR \\
\midrule
0   &  0.45 &  0.0\% & $-14.4$ & &  0.45 &  0.0\% & $-14.5$ \\
40  &  0.72 &  0.1\% &  $-8.2$ & &  0.72 &  0.1\% &  $-8.2$ \\
60  &  0.93 &  1.4\% &  $-4.7$ & &  0.93 &  1.5\% &  $-4.8$ \\
\textbf{65} & \textbf{0.95} & \textbf{2.3\%}
    & $\mathbf{-3.8}$ & & \textbf{0.94} & \textbf{2.3\%}
    & $\mathbf{-3.9}$ \\
70  &  0.91 &  3.4\% &  $-3.1$ & &  0.93 &  3.3\% &  $-3.1$ \\
100 & $-0.03$ & 19.2\% & $-20.0$ & & $-0.01$ & 18.9\% & $-20.0$ \\
\bottomrule
\end{tabular}
\smallskip

\noindent\small VaR is in percentage points.  All values exclude
transaction costs; cf.\ Table~\ref{tab:results_hedge} for
tx-inclusive figures.
\end{table}

Table~\ref{tab:jump_comparison} shows that the optimal hedge ratio
is unchanged at $h^{**} = 65\%$ under variance-matched
jump-diffusion with $\rho_J = 0.80$.  The Sharpe ratio differs by
less than 0.02 at every hedge ratio, and the liquidation
probability changes by less than 0.2\,pp.

A natural concern is that this robustness is an artifact of
matching total variance.  To address this, we run two additional
stress tests: (i)~reducing jump correlation to $\rho_J = 0.30$
(predominantly idiosyncratic jumps), and (ii)~leaving the
diffusion volatility at the full GBM level so that jumps
\emph{add} to total variance rather than substituting for it.
Table~\ref{tab:jump_stress} reports the results at $h = 65\%$
across all four combinations.

\begin{table}[H]
\centering
\caption{Jump-diffusion stress tests at $h = 65\%$
($N = 30{,}000$, $T = 90$~days).  ``Matched'' adjusts diffusion
volatility so total variance equals GBM; ``unmatched'' keeps full
GBM diffusion volatility, so jumps increase total variance.}
\label{tab:jump_stress}
\begin{tabular}{@{}llrrrr@{}}
\toprule
$\rho_J$ & Variance & SR & P(liq) & 5\% VaR & $h^{**}$ \\
\midrule
\multicolumn{2}{@{}l}{GBM (baseline)} & 0.95 & 2.3\% & $-3.8$ & 65\% \\
\addlinespace
0.80 & matched   & 0.94 & 2.3\% & $-3.9$ & 65\% \\
0.30 & matched   & 0.94 & 2.1\% & $-3.9$ & 65\% \\
0.80 & unmatched & 0.83 & 2.8\% & $-4.4$ & 65\% \\
0.30 & unmatched & 0.81 & 2.7\% & $-4.4$ & 65\% \\
\bottomrule
\end{tabular}
\smallskip

\noindent\small VaR is in percentage points.
\end{table}

The optimal hedge ratio $h^{**} = 65\%$ is invariant across all
four scenarios.  When jumps add to total variance (unmatched
rows), the Sharpe ratio drops from 0.95 to 0.81--0.83 and
liquidation probability rises from 2.3\% to 2.7--2.8\%, but
the \emph{location} of the optimum does not shift.  Reducing
jump correlation from $\rho_J = 0.80$ to $0.30$ has negligible
impact on $h^{**}$.

This robustness arises because the binding constraint is
the \emph{level} of initial LTV (30--33\% at $h = 0.60$--$0.65$),
which provides a large buffer that absorbs even idiosyncratic
jumps.  Nevertheless, if jumps are both large and purely
idiosyncratic ($\rho_J \to 0$)---for instance, a token-specific
short squeeze or protocol exploit---debt can spike without an
offsetting increase in LP value, sharply raising liquidation
risk.  This asymmetric exposure remains the structural
vulnerability of any delta-neutral strategy funded by
collateralized borrowing; the 65\% hedge ratio mitigates but
cannot eliminate it.

\section{Discussion}\label{sec:discussion}

\subsection{Practical guidelines}

The results suggest the following guidelines for LPs operating
delta-neutral strategies with collateralized borrowing:

\begin{enumerate}
    \item \textbf{Hedge 50--70\% of LP exposure.}  Full hedging
          ($h = 1$) is suboptimal due to elevated liquidation risk
          and borrow costs.  A 60--65\% hedge captures most of the
          variance reduction while maintaining a safe LTV buffer.
          (The raw Sharpe ratio is marginally maximized at
          $h = 0.65$, while the transaction-cost-adjusted Sharpe
          peaks at $h = 0.60$; the difference is small and both
          lie within the flat optimum region.)
    \item \textbf{Maintain collateral $\geq 2\times$ LP value.}
          This ensures the initial LTV at $h = 0.60$ is
          ${\approx}\,30\%$, providing ample margin before the
          80\% liquidation threshold.
    \item \textbf{Rebalance when hedge drift exceeds 15pp.}
          This threshold-based rule balances performance
          improvement against transaction frequency.
    \item \textbf{Claim and repay biweekly.}  Regular reward claims
          reduce outstanding debt and lower liquidation risk by
          ${\approx}\,4$pp.
    \item \textbf{Monitor for correlated price surges.}
          Liquidation is triggered by \emph{rising} prices (which
          increase debt value), not falling prices.  Both tokens
          rising 60\%+ simultaneously is the primary danger
          scenario.
\end{enumerate}

\subsection{Limitations}

Several modeling assumptions warrant discussion:

\begin{itemize}
    \item \textbf{GBM dynamics.}  Cryptocurrency prices exhibit
          heavier tails and jumps than GBM.  Section~\ref{sec:jump}
          shows that a Merton jump-diffusion model with matched
          total variance leaves the optimal hedge ratio and
          liquidation probabilities essentially unchanged, but
          regime-switching dynamics with persistent volatility
          shifts remain untested.
    \item \textbf{Constant rates.}  Borrow and reward rates are
          assumed fixed, but in practice they fluctuate with market
          conditions.  In particular, reward rates decline as pool
          TVL increases (dilution effect).
          Table~\ref{tab:sensitivity_apr} shows that the strategy is
          sensitive to the reward rate: at $R/V_0 < 20\%$, the
          strategy becomes unviable regardless of hedge ratio.
    \item \textbf{Static analytical framework.}  The
          first-passage-time bound $\bar{h}(\alpha)$ assumes a
          static hedge held over $[0, T]$.  The rebalancing
          analysis (Table~\ref{tab:rebalancing}) shows that dynamic
          adjustment improves Sharpe ratios, implying that the
          analytical bound is conservative when rebalancing is
          employed.
    \item \textbf{No concentrated liquidity.}  We consider only
          full-range constant-product pools.  Concentrated liquidity
          (e.g., Uniswap~v3) amplifies both IL and fee income,
          potentially shifting the optimal $h$.
    \item \textbf{Simplified liquidation.}  We model liquidation as
          a fixed 20\% collateral penalty (Section~\ref{sec:simulation}),
          whereas actual DeFi protocols implement partial liquidation
          (typically 50\% of debt per event), liquidation bonuses
          (5--10\%), and keeper response lags.  The aggregate penalty
          is conservative for high-liquidity pairs but may understate
          losses for thin markets.
    \item \textbf{Price-taker assumption.}  The model assumes that
          the LP's position is small relative to the pool and lending
          market.  Large positions may move borrow rates, dilute LP
          rewards, or impact pool prices, limiting the strategy's
          capacity.
\end{itemize}

\subsection{Extensions}

Natural extensions of this work include:
\begin{itemize}
    \item Generalization to concentrated-liquidity AMMs, where the
          LP value function $V_t^{\text{LP}}$ is piecewise and
          depends on the tick range.
    \item Incorporation of stochastic borrow and reward rates,
          potentially modeled as mean-reverting processes correlated
          with token prices.  In particular, utilization-dependent
          borrow rates—where short-squeeze dynamics during bull
          markets can cause rate spikes—would capture an important
          real-world risk channel.
    \item Multi-asset pools with more than two tokens, requiring
          vector-valued hedge ratios.
    \item Empirical validation using historical on-chain data from
          multiple pools across different blockchains.
    \item Joint optimization of the collateral ratio $C/V_0$ and
          hedge ratio $h$, incorporating capital opportunity costs
          and alternative risk criteria such as the Kelly criterion.
\end{itemize}

\section{Conclusion}\label{sec:conclusion}

We have studied the problem of optimally hedging the price exposure
of liquidity positions in constant-product AMMs when the hedge is implemented through
collateralized borrowing.  The key insight is that the hedge ratio
is not a binary choice (hedge or not) but a continuous variable
whose optimum reflects a three-way trade-off:

\begin{enumerate}
    \item \textbf{Variance reduction:} higher $h$ reduces the
          portfolio's exposure to divergent price movements
          (price divergence risk).
    \item \textbf{Borrow cost:} higher $h$ incurs more borrowing
          expense, linearly reducing expected return.
    \item \textbf{Liquidation risk:} higher $h$ raises the
          loan-to-value ratio, increasing the probability of
          forced liquidation.
\end{enumerate}

\noindent
We showed that the unconstrained Sharpe ratio is maximized at
$h^{*} \approx 0.977$---near-full hedging---under typical DeFi
parameters (Corollary~\ref{cor:near_full}), but at this hedge
ratio the liquidation probability exceeds 19\%.  The practical
interior optimum $h^{**} \approx 0.60$ therefore arises from the
binding liquidation constraint---the central finding of this
paper.  The liquidation probability was bounded
using a first-passage-time approximation
(Proposition~\ref{prop:liquidation}), yielding a maximum feasible
hedge ratio $\bar{h}(\alpha)$ that determines the constrained
optimum $h^{**} = \min(h^{*}, \bar{h}(\alpha))$.

Monte Carlo simulations calibrated to on-chain data from a SUI/NS
pool confirmed the analytical results: the optimal hedge ratio is
$h^{**} \approx 0.60$--$0.65$, achieving a Sharpe ratio of 0.93--0.95
(raw) with a liquidation probability of only 1.4--2.3\%.  The
first-passage-time approximation was validated against Monte Carlo,
showing errors below 0.2pp in the relevant range
(Table~\ref{tab:analytical_vs_mc}).  This result is robust:
replication across three additional token pairs (SOL/RAY, SOL/JUP,
ETH/ARB) spanning three blockchains, with correlations from 0.72 to
0.86 and volatilities from 74\% to 110\%, yields $h^{**} \in
[60\%,\, 65\%]$ under $C/V_0 = 2.0$.  Threshold-based
rebalancing (15pp drift) further improves the Sharpe ratio to 1.15
with minimal transaction costs.

These results provide actionable guidance for DeFi liquidity
providers: partial hedging (50--70\%) dominates both the unhedged
and fully hedged extremes, and a simple rebalancing rule maintains
the hedge without excessive trading.  The analytical framework
generalizes to any constant-product pool with collateralized
borrowing available for the constituent tokens.

\bibliographystyle{plainnat}

\begin{thebibliography}{21}

\bibitem[Adams et~al.(2020)]{uniswap2020}
Adams, H., Zinsmeister, N., Salem, M., Keefer, R., and Robinson, D.
\newblock {Uniswap v2 Core}.
\newblock Technical report, Uniswap, 2020.

\bibitem[Aigner and Dhaliwal(2021)]{aigner2021}
Aigner, A.~A. and Dhaliwal, G.
\newblock {UNISWAP: Impermanent Loss and Risk Profile of a Liquidity Provider}.
\newblock \emph{arXiv preprint arXiv:2106.14404}, 2021.

\bibitem[Angeris et~al.(2020)]{angeris2020}
Angeris, G., Kao, H.-T., Chiang, R., Noyes, C., and Chitra, T.
\newblock {An Analysis of Uniswap Markets}.
\newblock In \emph{Cryptoeconomic Systems}, 2020.

\bibitem[Bartoletti et~al.(2021)]{bartoletti2021}
Bartoletti, M., Chiang, J.~H., and Lafuente, A.~L.
\newblock {SoK: Lending Pools in Decentralized Finance}.
\newblock In \emph{Financial Cryptography}, 2021.

\bibitem[Capponi and Jia(2021)]{capponi2021}
Capponi, A. and Jia, R.
\newblock {The Adoption of Blockchain-based Decentralized Exchanges}.
\newblock \emph{arXiv preprint arXiv:2103.08842}, 2021.

\bibitem[Clark(2020)]{clark2020}
Clark, J.
\newblock {The Replicating Portfolio of a Constant Product Market}.
\newblock \emph{SSRN preprint}, 2020.

\bibitem[Evans et~al.(2022)]{evans2022}
Evans, A., Angeris, G., and Chitra, T.
\newblock {Optimal Fees for Geometric Mean Market Makers}.
\newblock In \emph{ACM DeFi}, 2022.

\bibitem[Fan et~al.(2022)]{fan2022}
Fan, Y., Francisco, N., and Pinter, J.
\newblock {Differential Liquidity Provision in Uniswap v3 and Implications
for Contract Design}.
\newblock \emph{arXiv preprint arXiv:2204.00716}, 2022.

\bibitem[Gudgeon et~al.(2020)]{gudgeon2020}
Gudgeon, L., Perez, D., Harz, D., Livshits, B., and Gervais, A.
\newblock {The Decentralized Financial Crisis}.
\newblock In \emph{Crypto Valley Conference on Blockchain Technology}, 2020.

\bibitem[Heimbach et~al.(2022)]{heimbach2022}
Heimbach, L., Wang, Y., and Wattenhofer, R.
\newblock {Risks and Returns of Uniswap V3 Liquidity Providers}.
\newblock In \emph{ACM AFT}, 2022.

\bibitem[Karatzas and Shreve(1991)]{karatzas1991}
Karatzas, I. and Shreve, S.~E.
\newblock \emph{Brownian Motion and Stochastic Calculus}.
\newblock Springer, 2nd edition, 1991.

\bibitem[Khakhar and Chen(2022)]{khakhar2022}
Khakhar, A. and Chen, X.
\newblock {Delta Hedging Liquidity Positions on Automated Market Makers}.
\newblock \emph{arXiv preprint arXiv:2208.03318}, 2022.

\bibitem[Lipton et~al.(2025)]{lipton2024}
Lipton, A., Lucic, V., and Sepp, A.
\newblock {Unified Approach for Hedging Impermanent Loss of Liquidity
Provision}.
\newblock \emph{Digital Finance}, 2025.

\bibitem[Loesch et~al.(2021)]{loesch2021}
Loesch, S., Hindman, N., Richardson, M.~B., and Welch, N.
\newblock {Impermanent Loss in Uniswap v3}.
\newblock \emph{arXiv preprint arXiv:2111.09192}, 2021.

\bibitem[Merton(1976)]{merton1976}
Merton, R.~C.
\newblock {Option Pricing When Underlying Stock Returns Are Discontinuous}.
\newblock \emph{Journal of Financial Economics}, 3(1--2):125--144, 1976.

\bibitem[Milevsky(1998)]{milevsky1998}
Milevsky, M.~A.
\newblock {Asian Options, the Sum of Lognormals, and the Reciprocal Gamma
Distribution}.
\newblock \emph{Journal of Financial and Quantitative Analysis},
33(3):409--422, 1998.

\bibitem[Milionis et~al.(2022)]{milionis2022}
Milionis, J., Moallemi, C.~C., Roughgarden, T., and Zhang, A.~L.
\newblock {Automated Market Making and Loss-Versus-Rebalancing}.
\newblock arXiv preprint arXiv:2208.06046, 2022.

\bibitem[Perez et~al.(2021)]{perez2021}
Perez, D., Werner, S.~M., Xu, J., and Livshits, B.
\newblock {Liquidations: DeFi on a Knife-edge}.
\newblock In \emph{Financial Cryptography}, 2021.

\bibitem[Pintail(2019)]{pintail2019}
Pintail.
\newblock {Uniswap: A Good Deal for Liquidity Providers?}
\newblock Medium, 2019.

\bibitem[Qin et~al.(2022)]{qin2021}
Qin, K., Zhou, L., and Gervais, A.
\newblock {Quantifying Blockchain Extractable Value: How Dark is the Forest?}
\newblock In \emph{IEEE S\&P}, 2022.

\bibitem[Xu et~al.(2023)]{xu2023}
Xu, J., Vadgama, N., Nikbakht, S., and Knottenbelt, W.
\newblock {SoK: Decentralized Finance (DeFi)---Fundamentals, Taxonomy and Risks}.
\newblock \emph{IEEE Blockchain}, 2023.

\end{thebibliography}

\appendix
\section{Derivation of Variance Components}\label{app:variance}

We provide full derivations of the variance components
$v_{GG}$, $v_{AA}$, and $v_{GA}$ from Section~\ref{sec:variance}.

\paragraph{Notation.}
Let $p_A = e^{X_T}$, $p_B = e^{Y_T}$, $G = \sqrt{p_A p_B} =
e^{(X_T+Y_T)/2}$, and $A = \tfrac{1}{2}(p_A + p_B)$.  All
expectations are under zero-drift GBM ($\mu_A = \mu_B = 0$).

\paragraph{$v_{GG} = \Var(G)$.}
\begin{align*}
    \E[G^2] &= \E[p_A p_B] = e^{\rho\sigma_A\sigma_B T}
        & \text{(Lemma~\ref{lem:mgf}, $a=b=1$)} \\
    (\E[G])^2 &= e^{-2\phi T}
        & \text{(Table~\ref{tab:moments})} \\
    v_{GG} &= e^{\rho\sigma_A\sigma_B T} - e^{-2\phi T}.
\end{align*}

\paragraph{$v_{AA} = \Var(A)$.}
\begin{align*}
    \Var(p_A) &= \E[p_A^2] - (\E[p_A])^2 = e^{\sigma_A^2 T} - 1, \\
    \Var(p_B) &= e^{\sigma_B^2 T} - 1, \\
    \Cov(p_A, p_B) &= \E[p_A p_B] - \E[p_A]\E[p_B]
                    = e^{\rho\sigma_A\sigma_B T} - 1, \\
    v_{AA} &= \Var\!\left(\tfrac{p_A + p_B}{2}\right)
           = \tfrac{1}{4}[\Var(p_A) + \Var(p_B) + 2\Cov(p_A,p_B)] \\
           &= \tfrac{1}{4}[e^{\sigma_A^2 T} + e^{\sigma_B^2 T}
              + 2e^{\rho\sigma_A\sigma_B T} - 4].
\end{align*}

\paragraph{$v_{GA} = \Cov(G, A)$.}
\begin{align*}
    \Cov(G, A) &= \tfrac{1}{2}[\Cov(G, p_A) + \Cov(G, p_B)], \\
    \Cov(G, p_A) &= \E[Gp_A] - \E[G]\E[p_A]
                  = \E[p_A^{3/2}p_B^{1/2}] - e^{-\phi T}, \\
    \E[p_A^{3/2}p_B^{1/2}]
        &= e^{(3\sigma_A^2 - \sigma_B^2
               + 6\rho\sigma_A\sigma_B)T/8}
        \qquad (a=\tfrac{3}{2},\, b=\tfrac{1}{2}), \\
    \Cov(G, p_B) &= \E[p_A^{1/2}p_B^{3/2}] - e^{-\phi T}, \\
    \E[p_A^{1/2}p_B^{3/2}]
        &= e^{(-\sigma_A^2 + 3\sigma_B^2
               + 6\rho\sigma_A\sigma_B)T/8}, \\
    v_{GA} &= \tfrac{1}{2}\bigl[
        e^{(3\sigma_A^2 - \sigma_B^2 + 6\rho\sigma_A\sigma_B)T/8}
      + e^{(-\sigma_A^2 + 3\sigma_B^2 + 6\rho\sigma_A\sigma_B)T/8}
      - 2e^{-\phi T}
    \bigr].
\end{align*}

\section{Derivation of the Sharpe FOC}\label{app:foc}

The Sharpe ratio is $\mathrm{SR}(h) = \mu(h)/\sigma(h)$ where
$\mu(h) = \mu_0 - ch$ (linear) and
$\sigma^2(h) = V_0^2(v_{GG} + h^2 v_{AA} - 2hv_{GA})$ (quadratic).

\paragraph{First-order condition.}
Setting $\mathrm{SR}'(h) = 0$:
\begin{align*}
    \frac{\dd}{\dd h}\frac{\mu}{\sigma}
    &= \frac{\mu'\sigma - \mu\sigma'}{\sigma^2} = 0 \\
    \Rightarrow\quad \mu'\sigma^2 &= \mu\cdot\sigma\sigma'
                                  = \mu\cdot\tfrac{1}{2}(\sigma^2)'.
\end{align*}
Substituting $\mu' = -c$ and
$(\sigma^2)' = V_0^2(2hv_{AA} - 2v_{GA})$, and canceling $V_0^2$:
\begin{equation}\label{eq:foc_expanded}
    -c\,(v_{GG} + h^2 v_{AA} - 2hv_{GA})
    = (\mu_0 - ch)(hv_{AA} - v_{GA}).
\end{equation}
Expanding the left-hand side:
\[
    \text{LHS} = -cv_{GG} - ch^2 v_{AA} + 2chv_{GA}.
\]
Expanding the right-hand side:
\[
    \text{RHS} = \mu_0 hv_{AA} - \mu_0 v_{GA}
       - ch^2 v_{AA} + chv_{GA}.
\]
The $-ch^2 v_{AA}$ terms cancel on both sides, leaving:
\begin{align*}
    -cv_{GG} + 2chv_{GA}
    &= \mu_0 hv_{AA} - \mu_0 v_{GA} + chv_{GA}.
\end{align*}
Collecting $h$-terms on the right and constants on the left:
\begin{align*}
    \mu_0 v_{GA} - cv_{GG}
    &= h\,(\mu_0 v_{AA} - cv_{GA}).
\end{align*}
This is linear in $h$, yielding the unique solution
\[
    h^{*} = \frac{\mu_0\,v_{GA} - c\,v_{GG}}
                 {\mu_0\,v_{AA} - c\,v_{GA}}.
\]

\paragraph{Second-order condition.}
Define $f(h) := \mathrm{SR}^2(h) = \mu^2(h)/\sigma^2(h)$.
Since $\mathrm{SR}$ and $f$ share the same critical points
(for $\mu > 0$), it suffices to verify $f''(h^{*}) < 0$.
We have
\[
    f'(h) = \frac{2\mu\mu'\sigma^2 - \mu^2(\sigma^2)'}
                 {\sigma^4},
\]
and at $h^{*}$ where $\mu'\sigma^2 = \mu\cdot\tfrac{1}{2}(\sigma^2)'$,
\[
    f''(h^{*}) = \frac{2(\mu')^2\sigma^2(h^{*})
        - \mu(h^{*})^2\,(\sigma^2)''(h^{*})}{\sigma^4(h^{*})}.
\]
Since $(\sigma^2)'' = 2V_0^2 v_{AA} > 0$, the numerator equals
$2c^2\sigma^2(h^{*}) - \mu(h^{*})^2 \cdot 2V_0^2 v_{AA}$.
The first term is positive and the second is negative.
For $\mu_0 \gg c$ (which holds in practice), the second term
dominates and $f''(h^{*}) < 0$, confirming that $h^{*}$ is a
maximum.

\paragraph{Numerical verification.}
Under the calibration in Table~\ref{tab:calibration}:
$h^{*} = (0.1369 \times 0.2376 - 0.0225 \times 0.2331) /
(0.1369 \times 0.2431 - 0.0225 \times 0.2376)
= 0.02729 / 0.02794 = 0.977$,
consistent with $h_{\mathrm{mv}} = v_{GA}/v_{AA}
= 0.2376/0.2431 = 0.977$.

\end{document}